\documentclass[12pt,draftclsnofoot,onecolumn]{IEEEtran}
\usepackage{amsmath,amsfonts,amsthm,amssymb}
\usepackage{algorithmic}
\usepackage{algorithm}
\usepackage{array}
\usepackage{float}
\usepackage{textcomp}
\usepackage{stfloats}
\usepackage{url}
\usepackage{verbatim}
\usepackage{graphicx}
\usepackage{subfigure}
\usepackage{cite}
\usepackage{color}
\usepackage{hyperref}
\usepackage{marvosym}

\makeatletter
\newcommand{\removelatexerror}{\let\@latex@error\@gobble}
\makeatother
\newcommand{\mbbE}{\mathbb{E}}

\newcommand{\mbXi}{\mathbf{\Xi}}

\newtheorem{theorem}{\hspace{-0.5em}\textbf{Theorem}}
\newtheorem{lemma}{\hspace{-0.5em}\textbf{Lemma}}

\newtheorem{corollary}{Corollary}

\begin{document}

	\title{Fast Capacity Estimation in  Ultra-dense Wireless Networks with Random Interference }

\author{Dandan Jiang,  Rui Wang, and Jiang Xue\textsuperscript{\Letter}, \textit{Member, IEEE}\thanks{		
Dandan Jiang,  Rui Wang and Jiang Xue are with School of Mathematics and Statistics, Xi'an Jiaotong University, Xi'an, China (e-mail: jiangdd@xjtu.edu.cn; 	wangrui\_math@stu.xjtu.edu.cn; x.jiang@xjtu.edu.cn). (Corresponding author:  Jiang Xue.)}}

%\thanks{The author Dandan Jiang was supported by Key technologies for coordination and inter-operation of power distribution service resource, Grant No. 2021YFB2401300, NSFC Grant No. 11971371 and the Fundamental Research Funds for the Central Universities.}
%\thanks{Manuscript received April 19, 2021; revised August 16, 2021.}

% The paper headers
%\markboth{Journal of \LaTeX\ Class Files,~Vol.~14, No.~8, August~2021}%

%\IEEEpubid{0000--0000/00\$00.00~\copyright~2021 IEEE}
% Remember, if you use this you must call \IEEEpubidadjcol in the second
% column for its text to clear the IEEEpubid mark.

\maketitle

\begin{abstract}
	
In wireless communication systems,  the accurate and reliable evaluation of channel capacity  is believed to be a fundamental and critical issue for terminals. 
However,  with the rapid development of wireless technology, large-scale communication networks  with significant random  interference have emerged, resulting in extremely high computational costs for  capacity calculation.
In  ultra-dense wireless networks with extremely large numbers of base stations (BSs) and users, 
we provide  fast estimation methods for  determining the capacity. We consider two scenarios  according to  the ratio of the number of users to the number of BSs, $\beta_m$. First, when  $\beta_m\leq1$, the  FIsher-Spiked Estimation (FISE) algorithm is proposed to determine the capacity   by modeling the channel matrix with  random interference as a  Fisher matrix. Second,
when $\beta_m>1$, based on a  closed-form expression for capacity estimation requiring solely simple computations, we prove that this estimation  stabilizes and remains invariant with increasing $\beta_m$.
Our methods can guarantee high accuracy on capacity estimation
 with low complexity, which is faster than  the existing methods.   Furthermore, our approaches exhibit excellent generality,  free of  network area shapes,  BS and user distributions, and sub-network  locations. 
Extensive simulation experiments  across various scenarios   demonstrate  the high accuracy and robustness of our methods. 

\end{abstract}


\section{Introduction}
	%\IEEEPARstart{T}{he}
	
The channel capacity, defined as the maximum achievable rate at which information can be transmitted through a channel, plays a significant role in wireless communication systems. Accurate and reliable evaluation of channel capacity is considered as a fundamental and critical issue for terminal performance. 
As  mobile communication technology improves by leaps and bounds,  increasingly intricate wireless systems emerge, posing the challenge of  determining the capacity with  low complexity and high accuracy.
As a pioneering work, Dr. Claude E. Shannon originally proposed the definition of channel capacity and provided its calculation method \cite{Shannon1948}, which is  known as the Shannon-Hartley theorem.  
Based on this theorem,  the capacity of  an additive white Gaussian noise (AWGN) channel with bandwidth $W$ can be  calculated by $C=W \log \left(1+P/(N_0 W)\right)$ \cite{cover2006elements},
where  $P$ denotes signal power and $N_0$ represents  noise power. 
Subsequently,  the growing demand for increased transmission capacity has propelled the development of multiple-input multiple-output (MIMO) antennas.
 As detailed in \cite{TE1999multi}, the capacity of  a multi-user (MU)-MIMO channel with $t$ transmission antennas and $r$ receiving antennas can be expressed as follows: 
\begin{equation} \label{mu-mimo}
	C=\mathbb{E}\left\{\log \operatorname{det}\left(\mathbf{I}+\frac{P}{t} \mathbf{H} \mathbf{H}^*\right)\right\},
\end{equation}
where $\mathbf{H}$ is the channel gain matrix, and $(\cdot)^*$ represents the Hermitian transpose. The notations $\mathbb{E}$ and $\operatorname{det(\cdot)}$ represent the expectation and the determinant of a matrix, respectively.  When $t$ and $r$ are large, the expression in \eqref{mu-mimo}  requires the determinant calculation of a large-dimensional matrix,  consuming substantial computational resources. Existing methods for direct computation of 
$C$ in \eqref{mu-mimo}, such as Cholesky decomposition and singular value decomposition (SVD), have a complexity of  $O(t^3)$, which   is unacceptable for future ultra-dense networks.

Recently, to further enhance the  capacity of  wireless channels, a capacity-centric (C$^2$) network architecture has been designed for future wireless communications\cite{yang2022}.   The C$^2$ architecture divides the whole network  into $M$ non-overlapping clusters,  each   operating  in parallel. 
Only the base stations (BSs) within each cluster collaborate  to serve nearby users, and there is non-negligible interference among different clusters. 
Its  average uplink capacity of the $m$-th cluster per BS can be obtained   by \cite{yang2022}
\begin{align} \label{Cm_real1}
	C_{m}=\mathbb{E}\left\{\frac{1}{J_{m}} \log \operatorname{det}\left(\mathbf{I}+P \boldsymbol{\Xi}_{m}^{-1} \mathbf{H}_{m} \mathbf{H}_{m}^{*} \right)\right\}, 
\end{align}
where $J_{m}$\footnote{In this paper, we  only consider the situation of single-antenna BS. For multi-antenna BS, $J_m$ can also represent the total number of antennas within the $m$-th cluster.} is the number of BSs in  the $m$-th cluster, $\mathbf{H}_{m}$ and $\boldsymbol{\Xi}_{m}$ denote the channel gain matrix and the noise-plus-interference matrix in this cluster, respectively. Additionally, $P \boldsymbol{\Xi}_{m}^{-1} \mathbf{H}_{m} \mathbf{H}_{m}^{*}$ represents the signal to the interference plus noise ratio (SINR) matrix of the $m$-th cluster.  
It can be seen that  \eqref{Cm_real1} is more general and more complicated than \eqref{mu-mimo}, since the C$^2$ network can  be  seen as a multiple MU-MIMO channel with inter-cluster interference. 
Moreover, we note that there are several other architectures in the literature that are built  similarly to C$^2$, such as clustered cell-free \cite{wang2022rcn}, CGN \cite{deng2022cgn}, and others \cite{Dai2017, 6G2021, wang2023}.
Their capacity expressions are consistent with formula \eqref{Cm_real1}, making the methods presented in this paper applicable to these architectures as well.

Many approaches have been developed to tackle capacity calculations involving large-dimensional channel matrices, diverging from conventional determinant-calculation-based methods.
\cite{TuVer2004}  utilized the random matrix theory (RMT) to  derive the asymptotic channel capacity with implicit expressions as  the number of BSs and  users increases, which are  mostly applicable to channels using code-division multiple access schemes, restricting their extension to other  models.
Moreover, several  works are specifically dedicated to determining the capacity  in  \eqref{Cm_real1}, such as  \cite{yang2022, TOSE2022, MPM2022}.
In  \cite{yang2022}, a closed-form estimation for $C_m$ was provided based on the convergence of the SINR matrix to a diagonal matrix under some asymptotic conditions. 
The TOSE method in \cite{TOSE2022} used the limiting spectral theory of spiked covariance matrix to achieve fast eigenvalues estimation, but
it relies on an assumption that the noise-plus-interference matrix $\boldsymbol{\Xi}_{m}$  is deterministic,
 inconsistent with the inherent randomness of $\boldsymbol{\Xi}_{m}$ in practical scenarios. In addition,
 %and converges to a diagonal matrix,
the accuracy of this method is also unstable for a channel matrix with extremely uneven signal descend.
%(one not real assumption and a not robust accracy) treated the noise-plus-interference matrix $\boldsymbol{\Xi}_{m}$ as deterministic,
The MPM method introduced in \cite{MPM2022} was developed   to  estimate the capacity by approximating the spectral distribution of the SIRN matrix in \eqref{Cm_real1}. 
Notably, all  of the above methods  (\cite{yang2022, TOSE2022, MPM2022})  rely on the assumption that the  matrix $\boldsymbol{\Xi}_{m}$ converges to a diagonal form as the number of users approaches infinity.
However, as discussed in \cite{couillet_liao_2022}, the convergence of large-dimensional matrices  is not equivalent to the convergence of their spectra. %Since the capacity formulas depend on the matrix spectrum, 
Therefore, these capacity estimations increasingly deviate from reality  as $\beta_m$ decreases, where $\beta_m$ represents the ratio of the number of users to the number of BSs in the $m$-th cluster.
Consequently, a method for determining the capacity that can  simultaneously guarantee  high accuracy, low complexity, and superior generality   needs to be further explored.

In this paper, we propose  fast estimation methods for  calculating capacity in ultra-dense wireless networks with random interference,  effectively avoiding  the  high complexity in conventional determinant-calculation-based methods. We consider two cases:  $\beta_m\leq1$ and $\beta_m>1$. When $\beta_m\leq1$, the   FIsher-Spiked Estimation (FISE) algorithm is proposed  to fast and accurately estimate $C_m$ in \eqref{Cm_real1}, which employs the limiting spectral theory of spiked Fisher matrix  to realize a fast estimation of eigenvalues.
%to estimate $C_m$ fast and accurately, which
%is based on the joint LSD and the estimation of the spikes.
When $\beta_m>1$, based on the capacity estimation $\widetilde C_m$  proposed by \cite{yang2022} and detailed later in formula \eqref{limit},  we further prove that this  estimation reaches a constant value that is independent of $\beta_{m}$.
The major contributions of this work can be concluded in terms of the following.
\begin{itemize}
  	\item[(a)]  The FISE algorithm is of low complexity with high accuracy in capacity estimation.  Different from the existing methods,  FISE eliminates the diagonal assumption of  $\mathbf \Xi_m$  and instead adopts its true structure. The proposed estimation  also preserves the randomness of the interference matrix, which aligns more closely with real-world case. 	 	
   In terms of computational complexity, when provided with the SINR matrix, the complexity for  FISE itself amounts to only $O(J_m)$, which is less than the existing methods.
   	\item[(b)] The  stability of the average cluster capacity  estimation for $\beta_m>1$ is proved.
   	Specifically, based on the  closed-form estimation  $\widetilde C_m$  provided by \cite{yang2022}, we demonstrate that this  estimation  is a constant that does not rely on the value of $\beta_{m}$. 
   	Thus, when  $\beta_m>1$, we do not need to repeat the capacity calculation and can instead use this stable constant, significantly saving computing power resources.
	\item[(c)]  Our proposed methods exhibit  remarkable generality, which are free of   nodes (BSs and users) distributions,  network area shapes, and the  locations of clusters within the network. In the simulation experiments, both square and round network regions were designed, along with two types of nodes distributions. Besides, three clusters with representative locations were chosen to calculate the channel capacity.  The experimental results indicate that the our methods perform  well across various scenarios. 

\end{itemize}

The arrangement of the following content is as below. First,  the system model is  formulated in Section \ref{isec2}.
Section  \ref{isec3} elaborates on
the procedures  of FISE to determine the capacity based on  the spiked Fisher matrix when $\beta_m\leq1$.  
Section  \ref{isec4} introduces  a computationally simple expression for capacity estimation when $\beta_m>1$ and proves its stability. As a by-product, an explicit  expression  of the average cluster capacity estimation per BS with a constant user density is also derived.
Section \ref{isec5} shows simulation results on the  capacity estimation comparison between  our proposed methods and other existing methods  
under different scenarios. %are shown.
Finally, Section \ref{isec6} draws the conclusion.
 
\textit{Notations:} We use the following notations throughout this paper:  Lower case, boldface lower case, and  boldface upper case letters represent scalars, vectors, and matrices, respectively, like $h$, $\mathbf h$, and $\mathbf{H}$. The $(i, j)$-th entry of $\mathbf{H}$ is denoted by $[\mathbf{H}]_{ij}$. Scripts such as $\mathcal{B}$ represent the sets.  Operator $\mathcal{C N}(\mu, \nu)$ is a complex Gaussian distribution with mean $\mu$ and covariance $\nu$. $\left\|\cdot\right\|_{F}$ and $|\cdot|$ represent the  Frobenius norm and the absolute value. Moreover, other operators $\operatorname{tr(\cdot)}$, $(\cdot)^*$, $\operatorname{det(\cdot)}$, $\mathbb{E}\{\cdot\}$, and $(\cdot)^{-1}$ represent the trace,  Hermitian transpose, determinant, expectation,  and matrix inverse, respectively.

\section{Modeling  the System with random interference} \label{isec2}
In this section, we  introduce the system model for C$^2$ networks proposed in \cite{yang2022} considering massive random interference.  The network is comprised  of two types of nodes: (1) $J$ single-antenna BSs (or access points in distributed-antenna systems \cite{How2016, Bashar2019}), which is denoted by $\mathcal{B}=\left\{b_{1}, b_{2}, \ldots b_{J}\right\}$; (2) $K$ single-antenna users, which is denoted by   $\mathcal{U}=$ $\left\{u_{1}, u_{2}, \ldots, u_{K}\right\}$. 
Suppose the entire network is decomposed into $M$ non-overlapping clusters, each  operating in parallel. The schematic diagram is  depicted in Fig. \ref{fig1}, where each color represents  a separate cluster,  circles represent  BSs, and  triangles represent   users. The cluster marked by the black pentagon is denoted as the $m$-th cluster,  used as a target to calculate the capacity. For the $m$-th cluster, let $\mathcal{C}_{m}$ represent the union of the sets of BSs and users in this cluster, thus $\bigcup_{m=1}^{M} \mathcal{C}_{m}=\mathcal{B} \cup \mathcal{U}$. Furthermore, denote the number of BSs  in $\mathcal{C}_{m}$ as $J_{m}$ and the number of users  as $K_{m}$. 
To reflect the ultra-dense scenario of the network, we assume that $J_{m}, K_{m} \rightarrow \infty.$

%\begin{figure}[htbp]
%	\centering
%	\vspace{-1em}
%	\includegraphics[width=0.7\textwidth]{Fig//fig1.eps}
%	\caption{Visualization of the wireless system model. Different colors represent different clusters. Circles are BSs, triangles are users. The cluster enclosed by the black circle is the  closest to the center of the network.}
%	\label{fig1}
%\end{figure}
\begin{figure}[htbp]
	\centering
	\vspace{-1em}
	\includegraphics[width=0.8\textwidth]{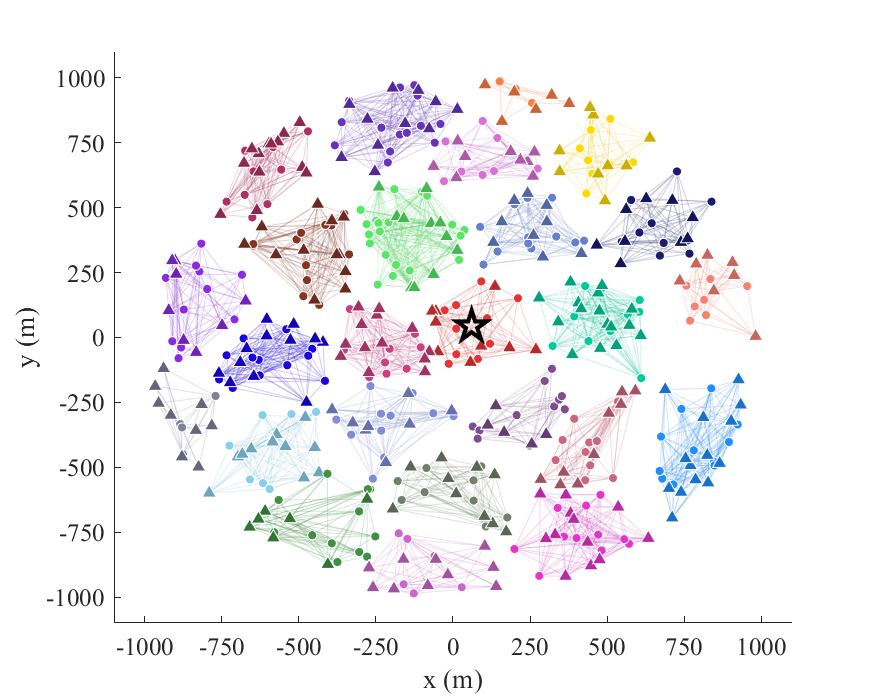}
	\caption{Schematic  diagram of the wireless network. Different colors represent different clusters. Circles are BSs, triangles are users. The cluster marked by the black pentagon is the  closest to the center of the network.}\label{fig1}
	\label{fig1_1}
\end{figure}

The channel gain between the BS $b_{j} \in \mathcal{C}_{m}$ and the user $u_{k} \in \mathcal{U}$ is modeled by
$$h_{m j k}=l_{m j k} g_{m j k},$$
where $g_{m j k} \sim \mathcal{C N}(0,1)$ is the small-scale fading, and
\begin{align}\label{lm}
	l_{m j k}= \begin{cases}d_{m j k}^{-1.75}, & d_{m j k}>d_{1} \\ d_{1}^{-0.75} d_{m j k}^{-1}, & d_{0}<d_{m j k} \leq d_{1} \\ d_{1}^{-0.75} d_{0}^{-1}, & d_{m j k} \leq d_{0}\end{cases}
\end{align} 
is the large-scale fading \cite{5G2020}. Here, $d_{m j k}$ is the Euclidean distance between the BS $b_{j} \in \mathcal{C}_{m}$ and the user $u_{k}$.  $d_{0}$ and $d_{1}$ are the near-field threshold and far-field threshold, respectively.

The uplink  signal model of the $m$-th cluster  is given by
\begin{align} \label{signal}
	\mathbf y_m=\sum_{u_{k} \in \mathcal{C}_{m}} \mathbf{h}_{m k}s_k+\sum_{u_{k} \notin \mathcal{C}_{m}} \mathbf{h}_{m k}s_k+\mathbf u_m,
\end{align} 
where  $\mathbf{h}_{m k}$ is a $J_m \times 1$ vector with $h_{mjk}$ as its $j$-th element, $s_k \sim \mathcal{C N}(0,P)$ is the signal of the user $u_k$, with $P$ being the transmit power of each user. Moreover,  $\mathbf u_m \sim \mathcal{C N}(0,N_0 \mathbf I)$ is the AWGN vector.
Define the channel gain matrix $\mathbf{H}_{m}$ as
$$
\mathbf{H}_{m}=\mathbf{L}_{m} \circ \mathbf{G}_{m},
$$ 
where $\mathbf{L}_{m}\in \mathbb{R}^{J_{m} \times K_{m}}$ is  the large-scale fading matrix, $\mathbf G_{m} \in \mathbb{C}^{J_{m}\times K_{m}}$ is the small-scaling fading matrix, with their $(j, k)$-th entries given by $\left[\mathbf{L}_{m}\right]_{j k}=l_{m j k}$ and $\left[\mathbf{G}_{m}\right]_{j k}=g_{m j k}$. Here, $\circ$ represents the Hadamard product. Besides, represent the noise-plus-interference matrix  as
\begin{equation} \label{NIR}
	\boldsymbol{\Xi}_{m}
	=N_{0} \mathbf{I}+P\sum_{u_{k} \notin \mathcal{C}_{m}} \mathbf{h}_{m k} \mathbf{h}_{mk}^{*}.
\end{equation}
Thus, the average capacity of cluster $m$ per BS can be obtained as referred in \cite{David2005, Shannon1948}
\begin{align}	\label{real}
	C_{m}&=\mathbb{E}\left\{\frac{1}{J_{m}} \log \operatorname{det}\left(\mathbf{I}+P \boldsymbol{\Xi}_{m}^{-1 / 2} \mathbf{H}_{m} \mathbf{H}_{m}^{*} \boldsymbol{\Xi}_{m}^{-1 / 2}\right)\right\} \notag \\
	&=\mathbb{E}\left\{\frac{1}{J_{m}} \log \operatorname{det}\left(\mathbf{I}+P \boldsymbol{\Xi}_{m}^{-1 / 2} (\mathbf{L}_{m} \circ \mathbf{G}_{m}) (\mathbf{L}_{m} \circ \mathbf{G}_{m})^{*} \boldsymbol{\Xi}_{m}^{-1 / 2}\right)\right\}.
\end{align}

In  existing works, people usually  consider only the randomness of the target channel matrix while regarding the interference matrix as fixed, which leads to the network system not adapting to real-time dynamic changes, and also brings  unstable capacity estimation errors. 
To address this issue,   our modeling process treats  both the target channel matrix and interference matrix as random matrices.
Define $\mathbf{B}_{m}=P \boldsymbol{\Xi}_{m}^{-1 / 2} (\mathbf{L}_{m} \circ \mathbf{G}_{m}) (\mathbf{L}_{m} \circ \mathbf{G}_{m})^{*} \boldsymbol{\Xi}_{m}^{-1 / 2}$. We can rewrite $\mathbf I+\mathbf{B}_{m}$ in the form of $\mathbf \Sigma_1\mathbf \Sigma_2^{-1}$, where $\mathbf\Sigma_1=\mathbf\Sigma_2+\mathbf\Delta$ and $\mathbf\Delta=\mathbf{B}_{m}\mathbf\Sigma_2$.
As a result, the  matrix $\mathbf I+\mathbf{B}_{m}$ in \eqref{real} is characterized as the ratio of two random covariance matrices, fitting the typical structure of Fisher matrices in RMT.
Furthermore,  $\mathbf I+\mathbf{B}_{m}$ always presents the dominant advantage of the first few eigenvalues, which is the so-called  spiked property.
To be specific, the matrix $\mathbf I+\mathbf{B}_{m}$ exhibits the characteristics of a spiked Fisher matrix as described in \cite{bai2010spectral}.
Therefore, in the following section, we will employ the spectral theory of the spiked Fisher matrix to develop a fast  method for estimating the  capacity based on the formula \eqref{real}.

%
%{\color{red}
%The SINR matrix in \eqref{real} is characterized as the ratio of two random covariance matrices, fitting the typical structure of Fisher matrices in RMT. Define $\mathbf{B}_{m}=P \boldsymbol{\Xi}_{m}^{-1 / 2} (\mathbf{L}_{m} \circ \mathbf{G}_{m}) (\mathbf{L}_{m} \circ \mathbf{G}_{m})^{*} \boldsymbol{\Xi}_{m}^{-1 / 2}$. We can rewrite $\mathbf I+\mathbf{B}_{m}$ in the form of $\mathbf \Sigma_1\mathbf \Sigma_2^{-1}$, where $\mathbf\Sigma_1=\mathbf\Sigma_2+\mathbf\Delta$ and $\mathbf\Delta=\mathbf{B}_{m}\mathbf\Sigma_2$.
%According to \cite{bai2010spectral}, the matrix $\mathbf I+\mathbf{B}_{m}$ exhibits the characteristics of a spiked Fisher matrix.
%In the following section, we will employ the spectral theory of spiked Fisher matrix to develop a fast  method for estimating the  capacity based on the formula in equation \eqref{real}. }

\section{ FIsher-Spiked Estimation for determining  the capacity} \label{isec3}
 Assume that  $K_m/J_m \rightarrow \beta_m$  as both $J_m$ and $K_m$ tend to infinity.  This section proposes a FIsher-Spiked  Estimation (FISE) algorithm to determine the average capacity  per cluster in \eqref{real} when $\beta_m\leq 1$,  leveraging the spectral theory of spiked Fisher matrices.
The capacity computation is divided into two parts. One part is achieved by the fast estimation of top  $R$ spiked eigenvalues, while the other part  is estimated by the limiting spectral distribution of the  remaining non-spiked eigenvalues.

\subsection{The approximate signal model} \label{transform}

As shown in \eqref{signal} and \eqref{real}, the uplink  signal model of the $m$-th cluster  involves Hadamard product  $\mathbf{L}_{m} \circ \mathbf{G}_{m}$. However, existing studies on the Hadamard product, as discussed in \cite{Girko2001, Hachem2007, Hachem2008, Silver2023}, often yield results that are implicit and cannot directly  apply to practical scenarios.
Therefore, based on Theorem 1 in \cite{TOSE2022}, we consider  approximating the signal model \eqref{signal} without the diagonal assumption of $\mathbf \Xi_m$ but  preserving its randomness. We replace the Hadamard product $\mathbf{L}_{m} \circ \mathbf{G}_{m}$ with the matrix product $ \widetilde{\mathbf{L}}_{m} \mathbf{G}_{m}$, where  
\begin{equation} \label{appro_mat}
	\widetilde{\mathbf{L}}_{m}=\operatorname{diag}\left(l_{m 1k}, \ldots, l_{m J_{m}k}\right),
\end{equation}
and $l_{m jk}=1/K_{m} \sum_{k=1}^{K_{m}} l_{m j k}$.   
This approximation  forms the foundation of the design of the FISE algorithm, and its optimality can be proven in the following lemma.  The proof  is similar to that of Theorem 1 in \cite{TOSE2022}, except for the target matrix  $\mathbf{L}_m$ instead of their $\mathbf{Q}_m = P^{1/2}\mathbf{\Xi}_m^{-1/2}\mathbf{L}_m$.
%. In \cite{TOSE2022}, the diagonal treatment of \eqref{appro_mat} is applied to the matrix $\mathbf{Q}_m = P^{1/2}\mathbf{\Xi}_m^{-1/2}\mathbf{L}_m$, whereas in this paper, it is applied only to $\mathbf{L}_m$.

\begin{lemma}\label{lem1}
	For any matrix $\breve{\mathbf{L}}_{m}$, define
	$$
	\boldsymbol{\Delta}_{m}=\mathbf{L}_{m} \circ \mathbf{G}_{m}-\breve{\mathbf{L}}_{m} \mathbf{G}_{m}.
	$$
	$\mathbb{E}\left(\left\|\boldsymbol\Delta_{m}\right\|_{F}^{2}\right)$ reaches the minimum value  if and only if $\breve{\mathbf{L}}_{m}$ takes the value of (\ref{appro_mat}), and the minimum is
	$$
	\mathbb{E}\left(\left\|\boldsymbol\Delta_{m}\right\|_{F}^{2}\right)|_{\min }=\sum_{j=1}^{J_m} \left[\sum_{k=1}^{K_m} l_{mjk}^2 - \frac{1}{K_m}\left(\sum_{k=1}^{K_m} l_{mjk}\right)^2\right].
	$$
\end{lemma}

Lemma \ref{lem1} demonstrates that the error between matrices $\mathbf{L}_{m} \circ \mathbf{G}_{m}$ and $ \breve{\mathbf{L}}_{m} \mathbf{G}_{m}$ is minimized when $\breve{\mathbf{L}}_{m}$ takes the value specified in (\ref{appro_mat}). By substituting $\mathbf{L}_{m} \circ \mathbf{G}_{m}$  with the matrix product $ \widetilde{\mathbf{L}}_{m} \mathbf{G}_{m}$, we shift our focus  to calculating the following  capacity approximation 
\begin{align}
    \widehat C_{m}=\mathbb{E}\left\{\frac{1}{J_{m}} \log \operatorname{det}\left(\mathbf{I}+P \boldsymbol{\Xi}_{m}^{-1 / 2} \widetilde{\mathbf{L}}_{m} \mathbf{G}_{m}\mathbf{G}_{m}^*\widetilde{\mathbf{L}}_{m}^* \boldsymbol{\Xi}_{m}^{-1 / 2} \right)\right\}.
	\label{appro} 
\end{align}

\subsection{FISE algorithm   for $\beta_m\leq 1$} \label{FISE}

For the  approximate capacity given in  \eqref{appro},  
we first concentrate on the scenario where $\beta_m\leq 1$, i.e., $K_m\leq J_m$. 
Define $\mathbf{P}_{m}=P \boldsymbol{\Xi}_{m}^{-1 / 2} \widetilde{\mathbf{L}}_{m} \mathbf{G}_{m}\mathbf{G}_{m}^*\widetilde{\mathbf{L}}_{m}^* \boldsymbol{\Xi}_{m}^{-1 / 2}$, and we can obtain the spectral decomposition of the matrix $\mathbf I+\mathbf P_m$ as
$$\mathbf I+\mathbf P_m=\mathbf U \mathbf \Lambda \mathbf U^*, $$
where $\mathbf \Lambda=\text{diag}(\lambda _1,\cdots,\lambda _{J_m})$ is a diagonal matrix with the descending eigenvalues of $\mathbf I+\mathbf P_m$, and 
%  eigenvalues $\lambda _1 \geq \lambda _2 \geq \cdots \geq \lambda _{J_m}$. 
 $\mathbf U$ is a $J_m \times J_m$ unitary matrix consisting of the  corresponding eigenvectors. %to the descending eigenvalues of $\mathbf I+\mathbf P_m$, and
 % Using this spectral decomposition, 
 Then, the capacity $\widehat C_m$ in \eqref{appro} can be determined by summing the logarithms of eigenvalues:
\begin{equation}\label{Cm_eig}
	\widehat C_m= \frac{1}{J_{m}} \mbbE \left\{\log \operatorname{det}\left(\mathbf I+\mathbf{P}_{m}\right) \right\}=\frac{1}{J_{m}} \sum_{j=1}^{J_m} \mbbE \left\{\log \left(\lambda_{j}\right)\right\}.
\end{equation}

As  mentioned in Section \ref{isec2}, the matrix $\mathbf I+\mathbf{P}_{m}$ can also be rewritten as $\mathbf \Sigma_1\mathbf \Sigma_2^{-1}$, where $\mathbf\Sigma_1=\mathbf\Sigma_2+\mathbf\Delta$ and $\mathbf\Delta=\mathbf{P}_{m}\mathbf\Sigma_2$. Therefore,  the matrix $\mathbf I+\mathbf{P}_{m}$  is also of the spiked Fisher matrix type, when  $\mathbf{P}_{m}$ is a low rank matrix.
 We denote the rank of the matrix $\mathbf P_m$ as $R$, then $R \leq \min(J_m,K_m)$.  Further, $\text{rank}(\mathbf\Delta)=\text{rank}(\mathbf{P}_{m})=R$. If   $R$ is  small compared  to $J_m$,  $\mathbf \Sigma_1\mathbf \Sigma_2^{-1}$ is a standard  spiked Fisher matrix that has been extensively studied in the literature, such as \cite{wang2016, Jiang2021}. 
If  $R$ diverges as $J_m$ approaches infinity but still exhibits a few dominant eigenvalues, $\mathbf\Sigma_1\mathbf \Sigma_2^{-1}$  
remains  a spiked Fisher matrix with a diverging number of spikes,  as discussed  in \cite{xie2021} and  \cite{zheng2023}.

For each  observation of  $\mathbf I+\mathbf{P}_{m}$ in practical scenarios,  we can similarly rewrite it in the form of  $\mathbf{S}_1\mathbf{S}_2^{-1}$, where 
 $\mathbf{S}_1$ and $\mathbf{S}_2$ are regarded as the corresponding sample covariance matrices of $\mathbf \Sigma_1$ and $\mathbf \Sigma_2$,  respectively. Thus, we study the sample limiting properties of $\mathbf{S}_1\mathbf{S}_2^{-1}$ to capture the underlying population information. 
Denote the descending eigenvalues of $\mathbf{S}_1\mathbf{S}_2^{-1}$ by $\rho_{1}\geq \rho_{2} \geq \cdots \geq \rho_{J_m}$, then the capacity in \eqref{Cm_eig} can be further approximated as
\begin{align}\label{Cm_appro}
	\widehat C_m \approx \frac{1}{J_{m}} \sum_{j=1}^{J_m}  \mbbE \left\{\log \rho_{j}\right\}.
\end{align}

Distinguishing from \cite{TOSE2022}, which requires selecting the number of  spiked eigenvalues by adjusting parameters in advance, this study is free of  tuning parameters.
We just assume that the number of spikes in $\mathbf S_1\mathbf S_2^{-1}$  is exactly $R$, the rank  of the matrix $\mathbf{P}_{m}$.  According to \cite{bai2010spectral}, as $J_{m}, K_{m}\rightarrow \infty$, $K_m/J_m\rightarrow \beta_m\in(0,1)$, and $J_m/(K-K_m)\rightarrow y_m\in(0,1) $, the limiting spectral distribution (LSD) generated from the $J_m-R$ non-spiked eigenvalues of $\mathbf S_1\mathbf S_2^{-1}$ has the following density function
\begin{equation} \label{lsd}
	p_{\beta_m, y_m}(x)= \begin{cases}\displaystyle\frac{\beta_m\left(1-y_{m}\right) }{2 \pi x\left(1+\beta_m  y_{m}x\right)}\sqrt{(b_m-x)(x-a_m)}, &  a_m\leq x\leq b_m, \\ 0, & \text { else},\end{cases}
\end{equation}
and has a point mass $1-\beta_m$ at the origin, where $a_m,b_m=(1\mp\mu)^2/(1-y_{m})^2$ with $\mu=\sqrt{(1+\beta_m y_m-y_m)/\beta_m}$. 
As discussed in \cite{bai2010spectral}, the top $R$ spiked eigenvalues are located outside the support set $[a_m, b_m]$, while the rest $J_m-R$ non-spiked eigenvalues lie within the interval $[a_m, b_m]$.
Therefore, the capacity expression in \eqref{Cm_appro}  can be further split into two terms as
\begin{eqnarray} \label{split}
	\widehat C_m &\approx& \frac{1}{J_{m}} \mbbE \left\{\sum_{j=1}^{R} \log \rho_{j}\right\}+\frac{1}{J_{m}} \mbbE \left\{\sum_{k=R+1}^{J_m} \log \rho_{k}\right\}
	\triangleq C_{m1} +C_{m2}.
\end{eqnarray} 
%Write $\widehat C_{m1}=1/{J_{m}} \sum_{j=1}^{R} \log \rho_{j}$ and $\widehat C_{m2}=1/{J_{m}} \sum_{k=R+1}^{J_m} \log \rho_{k}$.  In the following,
Then, we will  provide the estimates  for $C_{m1}$ and $C_{m2}$  separately.

We first concentrate on estimating $C_{m1}$. As aforementioned, the spiked eigenvalues $\rho_{1} \geq \rho_{2} \geq \cdots \geq \rho_{R}$ are outside of the support set $[a_m, b_m]$ as $J_{m}, K_{m}\rightarrow \infty$ and $K_m/J_m\rightarrow \beta_m$. These eigenvalues are expected to be greater than 
$b_m$, thus   $b_m$ can be regarded as their lower bound, that is,  $\rho_{1} \geq \rho_{2} \geq \cdots \geq \rho_{R}>b_m$. 
For simple calculation, we 
suppose that $\rho_{1}, \rho_{2},\cdots, \rho_{R}$ are evenly spaced outside the support set with  an equal step of $\Delta \rho$, then  we can obtain their estimates as
\begin{equation}\label{eig_appro}
	\hat\rho_{j}=b_m+(R+1-j) \Delta \rho, j=1,2, \cdots, R.
\end{equation}
These eigenvalues should satisfy
\begin{equation}\label{condition}
	\sum_{j=1}^{R} \hat\rho_{j} \approx R+\operatorname{tr} (\mathbf{P}_{m}).
\end{equation}
Combining the formulas \eqref{eig_appro} and \eqref{condition}, we can obtain the value of the step $\Delta \rho$ as
\begin{equation}\label{space}
	\Delta \rho=\frac{2\left(\operatorname{tr} (\mathbf{P}_{m})+R-R b_m\right)}{R(R+1)}. 
\end{equation}
By substituting \eqref{space}  into  \eqref{eig_appro},  the values of  all approximate spiked eigenvalues $\hat\rho_{j}, j=1,\cdots, R$ can be  obtained. Consequently, $C_{m1}$ in (\ref{split}) can be  estimated by
\begin{equation}\label{Cm1}
	\widehat C_{m1} = \frac{1}{J_{m}} \sum_{j=1}^{R} \log \hat\rho_{j}.
\end{equation}

Next, to estimate  $C_{m2}$, we utilize the LSD of the Fisher matrix as defined in \eqref{lsd}. 
The term $C_{m2}$ is calculated by the $J_m-R$ non-spiked eigenvalues $\rho_{R+1},\cdots,\rho_{J_m}$, which are located within the support set of the LSD. Furthermore, since $\mathbf{P}_{m}$ ia a non-negative definite matrix, the eigenvalues of  $\mathbf I+\mathbf{P}_{m}$
 should be greater than 1. Then, by  the spectra convergence in the Fisher matrix, $C_{m2}$ can be approximated by the following integral
\begin{equation}\label{Cm2}
	\widehat C_{m2} =\int_{\max(1,a_m)}^{b_m} \log(x)p_{\beta_m, y_{m}}(x)dx.
\end{equation}

Finally, an estimation of the capacity in \eqref{Cm_appro}  is obtained by  combining the results from \eqref{Cm1} and \eqref{Cm2}. %Naturally, it  provides a good estimate for $C_m$.

The algorithm for determine $C_m$ for $\beta_m \leq 1$  is summarized as below.
\begin{algorithm}[htbp]
	\caption{FIsher-Spiked Estimation (FISE)}
	\begin{algorithmic}
		\STATE 
		\STATE \textbf{Input:} $\mathbf{P}_{m}, K, J_{m}, K_{m}$, $R$. 
		\STATE \textbf{Output:} Estimation of ${C}_{m}$.
		\STATE \hspace{0.5cm} 1: Calculate $\mathrm{tr} (\mathbf{P}_{m})$.
		\STATE \hspace{0.5cm} 2: Calculate $\Delta \rho={2\left[\operatorname{tr} (\mathbf{P}_{m})+R-R b_m\right]}/{[R(R+1)]}$, where $R$ is the rank of $\mathbf{P}_{m}$ and $b_m$ is the right endpoint of the support set in the density function in \eqref{lsd}.
		%=(1+\mu)^2/(1-y_{m})^2$.
		\STATE \hspace{0.5cm} 3: Compute $\hat\rho_{j}=b_m+(R+1-j) \Delta \rho, j=1, \cdots, R$.
		\STATE \hspace{0.5cm} 4: Compute $\widehat{C}_{m1}$ according to \eqref{Cm1}.
		\STATE \hspace{0.5cm} 5: Compute $\widehat{C}_{m2}$ according to \eqref{Cm2}.
		\STATE \hspace{0.5cm} 6: Add $\widehat{C}_{m2}$ to $\widehat{C}_{m1}$, and the sum is used as the estimation of $C_m$.
	\end{algorithmic}
	\label{alg2}
\end{algorithm}

It is worth mentioning that,  according to the algorithm procedure,  the time complexity of the FISE algorithm amounts to $O(J_m)$, which is equivalent to the complexity of  TOSE, lower than the complexity of the MPM algorithm  $O(J_m^2)$, and much lower than that of determinant-calculation-based methods  $O(J_m^3)$.

\section{Capacity stability for $\beta_m>1$} \label{isec4}
This section introduces a simplified capacity estimate that requires only basic computations when $\beta_m>1$.  We then theoretically prove that this estimation  is a constant value, which is free  of the parameter $\beta_m$.

The Lemma 1 in \cite{yang2022} states that the matrix $\mbXi_m^{-1/2}\mathbf H_m\mathbf H_m^*\mbXi_m^{-1/2}$ converges to a positive definite diagonal matrix  ${\mathbf R}_m$ as  $K_m$ approach infinity, where
$$
{\mathbf R}_m =\operatorname{diag}\left(r_{11}, \ldots, r_{jj}, \ldots, r_{J_m J_m}\right)
$$
with
$$r_{jj}^m=\frac{\sum_{ u_k\in \mathcal{C}_{m}}l_{mjk}^2}{N_0+P\sum_{ u_k\notin \mathcal{C}_{m}} l_{mjk}^2}.$$

However, as mentioned in the Introduction, for small values of $\beta_m<1$,  even if a large-dimensional matrix converges to a  limiting matrix,   its corresponding spectrum does not converge to  that of the  limiting matrix. Thus, the  diagonal limiting assumption for   the matrix $\mbXi_m^{-1/2}\mathbf H_m\mathbf H_m^*\mbXi_m^{-1/2}$ results in significant errors in capacity estimation. But when $\beta_m>1$, the error between the two spectra decreases as $\beta_m$ increases.
Therefore, the following expression can be regarded as a reasonable estimate of $C_m$ for large $\beta_m>1$.
\begin{align}
	\widetilde C_m&=\lim_{K_m\rightarrow\infty}\mathbb{E}\left\{\frac{1}{J_m}\log_2 \operatorname{det}(\mathbf I+P{\mathbf R}_m )\right\} \notag\\
	&=\lim_{K_m\rightarrow\infty}\frac{1}{J_m}\sum_{j=1}^{J_m}\log_2\left(1+\frac{P\sum_{ u_k\in \mathcal{C}_{m}}l_{mjk}^2}{N_0+P\sum_{ u_k\notin \mathcal{C}_{m}} l_{mjk}^2}\right). \label{limit}
\end{align}
Hence, the expression \eqref{limit} provides a straightforward approach to computing the capacity, reducing the complexity from high-dimensional matrix operations to simple numerical calculations.

When calculating capacity for the different values of % the  relatively large  
$\beta_m$, the method given in  \cite{yang2022} involves performing a repeated calculation  for each $\beta_m$ according to formula \eqref{limit}. However,  we find that the capacity  estimate \eqref{limit} is actually a stable value, and we can utilize the property thereby avoiding  repeated calculations. 
To prove this in detail, we introduce the following transformation.
%With the asymptotic  
Under the assumption of $K_m \rightarrow \infty$, 
the expression of $r_{jj}^m$ can be  transformed by replacing the discrete 
distributions of the network nodes with continuous density as follows.
$$
\lim_{K_m\rightarrow\infty}\frac{\sum_{ u_k\in \mathcal{C}_{m}}l_{mjk}^2}{N_0+P\sum_{ u_k\notin \mathcal{C}_{m}} l_{mjk}^2}=\frac{\int_{\mathbf{y} \in \mathcal{D}_m} f(\mathbf{x}-\mathbf{y}) \rho_u(\mathbf{y}) d \mathbf{y}}{N_0+P \int_{\mathbf{y} \in \mathcal{D}_0 \backslash \mathcal{D}_m} f(\mathbf{x}-\mathbf{y}) \rho_u(\mathbf{y}) d \mathbf{y}}, \quad\left(m=1,2, \ldots, J_m\right),
$$
where $\mathcal{D}_0$ denotes the two-dimensional region spanning the entire network, and  $\mathcal{D}_m \subseteq \mathcal{D}_0$ denotes the region spanned by the $m$-th cluster. Additionally,  $\mathbf{x}$ and $\mathbf{y}$ represent the location coordinates of the BSs and users, respectively,  $\rho_u(\mathbf{y})$ is the density function of users distributed over $\mathcal{D}_0$, and $f(\mathbf{x}-\mathbf{y})=\gamma d^{-\epsilon}_{\mathbf{x}\mathbf{y}}$, where    
\begin{align*}
	\gamma= \begin{cases}1, ~ d_{\mathbf{x y}}>d_1, \\
		d_1^{-1.5}, ~ d_0<d_{\mathbf{x y}} \leq d_1, \\
		d_1^{-1.5} d_0^{-2}, ~ 0<d_{\mathbf{x y}} \leq d_0,\end{cases}  
	\epsilon= \begin{cases}3.5,~  d_{\mathbf{x y}}>d_1. \\
		2,  ~d_0<d_{\mathbf{x y}} \leq d_1, \\
		0, ~ 0<d_{\mathbf{x y}} \leq d_0.\end{cases}
\end{align*}
Thus, the  average cluster capacity estimation per BS is given by \cite{yang2022} 
\begin{equation} \label{int_limit}
	\widetilde C_m=\log \frac{\frac{N_0}{P}+\int_{\mathbf{y} \in \mathcal{D}_0} f\left(\mathbf{x}_{m j}-\mathbf{y}\right) \rho_u(\mathbf{y}) d \mathbf{y}}{\frac{N_0}{P}+\int_{\mathbf{y} \in \mathcal{D}_0 \backslash \mathcal{D}_m} f\left(\mathbf{x}_{m j}-\mathbf{y}\right) \rho_u(\mathbf{y}) d \mathbf{y}},
\end{equation}
for some BS $b_j$ in the $m$-th cluster with coordinate denoted by 
$\mathbf{x}_{m j}$.

By  replacing the discrete distributions of the network nodes with continuous density and focusing on the interference-limited regime, where the background noise is negligible, we can prove that the average cluster capacity estimation in \eqref{limit}  is a stable value that remains invariant with respect to  $\beta_m$, as described in the following theorem.
\begin{theorem}\label{th1}
	 For ultra-dense wireless networks,  the average cluster capacity estimate per BS $\widetilde C_m(\beta_{m})$ is a constant value that is independent of $\beta_{m}(>1)$, where $\widetilde C_m(\beta_{m})$ is calculated using formula \eqref{limit}  corresponding to the ratio $\beta_m$. Specifically,  for any two values $\beta_{m1}\neq\beta_{m2}$ and $\beta_{m1},~ \beta_{m2}>1$,   the conclusion that $\widetilde C_m(\beta_{m1})=\widetilde C_m(\beta_{m2})$ holds.
\end{theorem} 

\begin{proof}
	If no confusion, represent $\rho_{\beta_m}(\mathbf{y})$ as the value of $\rho_{u}(\mathbf{y})$ at $\beta_m$. Then the average cluster capacity estimate at $\beta_{m1}>1$  can be obtained as
	\begin{align} \label{beta_0}
		\widetilde C_m(\beta_{m1})=\log \frac{\frac{N_0}{P}+\int_{\mathbf{y} \in \mathcal{D}_0} f\left(\mathbf{x}_{m j}-\mathbf{y}\right) \rho_{\beta_{m1}}(\mathbf{y}) d \mathbf{y}}{\frac{N_0}{P}+\int_{\mathbf{y} \in \mathcal{D}_0 \backslash \mathcal{D}_m} f\left(\mathbf{x}_{m j}-\mathbf{y}\right) \rho_{\beta_{m1}}(\mathbf{y}) d \mathbf{y}}.
	\end{align}
	Since 	\begin{align} \label{int}
		\int_{\mathbf{y} \in \mathcal{D}_m} \rho_{\beta_m}(\mathbf{y}) d \mathbf{y}=K_m=\beta_m J_m,
	\end{align}
	thus $\int_{\mathbf{y} \in \mathcal{D}_m} \rho_{\beta_m}(\mathbf{y}) /\beta_m d \mathbf{y}= J_m$. By setting $J_m$ to a large and fixed value, we isolate the quantity $\beta_m$ to study its impact on network capacity $\widetilde C_m(\beta_{m})$.  Therefore, 
	when the distribution of user nodes remains unchanged, an increase in $\beta_m$ will lead to  a proportional increase in density $\rho_{\beta_m}(\mathbf{y})$, i.e.
	$$
	\frac{\rho_{\beta_{m2}}(\mathbf{y})}{\rho_{\beta_{m1}}(\mathbf{y})}=\frac{\beta_{m2}}{\beta_{m1}},  
	$$
	where $\beta_{m2}\neq\beta_{m1}$ and $\beta_{m1},~ \beta_{m2}>1$.  Therefore,  $\widetilde C_m(\beta_{m2})$ can be calculated as	
	\begin{align}
		\widetilde C_m(\beta_{m2})&=\log \frac{\frac{N_0}{P}+\int_{\mathbf{y} \in \mathcal{D}_0} f\left(\mathbf{x}_{m j}-\mathbf{y}\right) \rho_{\beta_{m2}}(\mathbf{y}) d \mathbf{y}}{\frac{N_0}{P}+\int_{\mathbf{y} \in \mathcal{D}_0 \backslash \mathcal{D}_m} f\left(\mathbf{x}_{m j}-\mathbf{y}\right) \rho_{\beta_{m2}}(\mathbf{y}) d \mathbf{y}} \notag\\
		&=\log \frac{\frac{N_0}{P}+\frac{\beta_{m2}}{\beta_{m1}}\int_{\mathbf{y} \in \mathcal{D}_0} f\left(\mathbf{x}_{m j}-\mathbf{y}\right) \rho_{\beta_{m1}}(\mathbf{y}) d \mathbf{y}}{\frac{N_0}{P}+\frac{\beta_{m2}}{\beta_{m1}}\int_{\mathbf{y} \in \mathcal{D}_0 \backslash \mathcal{D}_m} f\left(\mathbf{x}_{m j}-\mathbf{y}\right) \rho_{\beta_{m1}}(\mathbf{y}) d \mathbf{y}}\notag \\
		&=\log \frac{\frac{N_0}{P}\frac{\beta_{m1}}{\beta_{m2}}+\int_{\mathbf{y} \in \mathcal{D}_0} f\left(\mathbf{x}_{mj}-\mathbf{y}\right) \rho_{\beta_{m1}}(\mathbf{y}) d \mathbf{y}}{\frac{N_0}{P}\frac{\beta_{m1}}{\beta_{m2}}+\int_{\mathbf{y} \in \mathcal{D}_0 \backslash \mathcal{D}_m} f\left(\mathbf{x}_{m j}-\mathbf{y}\right) \rho_{\beta_{m1}}(\mathbf{y}) d \mathbf{y}} \label{beta_1}
	\end{align}
	When  the background noise is ignored, the expressions \eqref{beta_0} and \eqref{beta_1} are equivalent, i.e. 
	$$
	\widetilde C_m(\beta_{m2})=\widetilde C_m(\beta_{m1}).
	$$
	The proof is completed.
\end{proof}

Theorem \ref{th1} implies that the estimated  capacity $\widetilde C_m$ stabilizes and does not vary with further increases in  $\beta_m$. Therefore, for  $\beta_m>1$, the estimation of the capacity $\widetilde C_m$ can be calculated only once using the simplified formula in \eqref{limit}, avoiding the repeated calculations for each $\beta_m$. % and thus reducing computational overhead. 
Then, this estimation will approach the real value of the capacity $C_m$ with increasing $\beta_m$, as illustrated by  the experimental study detailed in Section \ref{isec4}. 
This is significant for simplifying the analysis and optimization of ultra-dense networks   where $\beta_m$ can vary widely. 
%Furthermore, while the capacity in \eqref{real} tends to a stable value for $\beta_m\gg 1$, noteworthy insights emerge from
%
%These results reveal that even for relatively moderate $\beta_m$ values ($\beta_m \geq 3$), the estimated performance surpasses expectations, yielding almost the same capacity as the baseline result. Such observations underscore the practical significance of the properties elucidated by Theorem \ref{th1}.

According to \eqref{int}, the explicit relationship between $\rho_{\beta_m}(\mathbf{y})$ and $\beta_m$ is difficult to obtain.  Consequently, deriving an explicit expression for $\widetilde{C}_m$ in \eqref{int_limit}  is formidable.
However, when the user nodes are uniformly distributed, a specific limit expression of $\widetilde C_m$ can be determined by the following corollary under the more general condition than in Theorem \ref{th1}, i.e., without ignoring background noise.

\begin{corollary}\label{coro1}
	Assuming that the user density is a constant and the numbers of BSs and users approach infinity, the average cluster capacity per BS is estimated by 
	\begin{equation}\label{Cm_uniform}
		\widetilde C_m=\log \frac{\int_{\mathbf{y} \in \mathcal{D}_0} f\left(\mathbf{x}_{m j}-\mathbf{y}\right)  d \mathbf{y}}{\int_{\mathbf{y} \in \mathcal{D}_0 \backslash \mathcal{D}_m} f\left(\mathbf{x}_{mj}-\mathbf{y}\right) d \mathbf{y}}.
	\end{equation}
	
\end{corollary} 
\begin{proof}
	When the  user density is a constant,  we simplify $\rho_{u}(\mathbf{y})$ as $\rho_{u}$,  which can be directly calculated as
	$$
	\rho_{u}=\frac{K_m}{|\mathcal{D}_m |}.
	$$
	Here, 	$|\mathcal{D}_m |$  denotes the area of $\mathcal{D}_m$ with a slight abuse of notation. Then, the  average cluster capacity per BS is estimated by
	\begin{align} 
		\widetilde C_m&=\log \frac{\frac{N_0}{P}+\int_{\mathbf{y} \in \mathcal{D}_0} f\left(\mathbf{x}_{m j}-\mathbf{y}\right) \rho_{u}(\mathbf{y}) d \mathbf{y}}{\frac{N_0}{P}+\int_{\mathbf{y} \in \mathcal{D}_0 \backslash \mathcal{D}_m} f\left(\mathbf{x}_{m j}-\mathbf{y}\right) \rho_{u}(\mathbf{y}) d \mathbf{y}}  \notag\\
		&=\log \frac{\frac{N_0}{P}+\frac{K_m}{|\mathcal{D}_m |}\int_{\mathbf{y} \in \mathcal{D}_0} f\left(\mathbf{x}_{m j}-\mathbf{y}\right)  d \mathbf{y}}{\frac{N_0}{P}+\frac{K_m}{|\mathcal{D}_m |}\int_{\mathbf{y} \in \mathcal{D}_0 \backslash \mathcal{D}_m} f\left(\mathbf{x}_{mj}-\mathbf{y}\right) d \mathbf{y}} \notag\\
		&=\log \frac{\frac{N_0}{P}\frac{|\mathcal{D}_m |}{K_m}+\int_{\mathbf{y} \in \mathcal{D}_0} f\left(\mathbf{x}_{m j}-\mathbf{y}\right)  d \mathbf{y}}{\frac{N_0}{P}\frac{|\mathcal{D}_m |}{K_m}+\int_{\mathbf{y} \in \mathcal{D}_0 \backslash \mathcal{D}_m} f\left(\mathbf{x}_{mj}-\mathbf{y}\right) d \mathbf{y}}. \label{cap}
	\end{align}
	Even if  the background noise $N_0$ is not negligible,  the term $N_0|\mathcal{D}_m |/(PK_m)$ still tends to 0 due to the asymptotic assumption of $K_m \rightarrow \infty$. The expression \eqref{cap} can thus  be  written as
	\begin{align} 
		\widetilde C_m=\log \frac{\int_{\mathbf{y} \in \mathcal{D}_0} f\left(\mathbf{x}_{mj}-\mathbf{y}\right)  d \mathbf{y}}{\int_{\mathbf{y} \in \mathcal{D}_0 \backslash \mathcal{D}_m} f\left(\mathbf{x}_{mj}-\mathbf{y}\right) d \mathbf{y}}.
	\end{align}
	The proof is completed.
\end{proof}

We note that in \cite{yang2022}, the same result as \eqref{Cm_uniform} was derived under the assumption of negligible background noise. Conversely, our work establishes that the conclusion of Corollary \ref{coro1} holds even under the more general condition where background noise is considered.

\section{Performance evaluation}\label{isec5}
 In this section, numerical simulations are conducted to assess the high efficiency and exceptional generality of our proposed methods for capacity determination.

\subsection{Network settings}
In our simulations, we explore  two kinds of ultra-dense wireless network scenarios. The detailed designs of these scenarios are as follows, with visualizations provided in Fig. \ref{fig_net}.
\begin{enumerate}
	\item[S1.] The network is circular with a diameter of $2D$.  The locations of BSs and users  obey a homogeneous Poisson point process (PPP) with intensity $\Lambda_b$ and $\Lambda_u$ \cite{Khan2015}.   Therefore, $J$ and $K$ follow  Poisson distributions:
	$$\text{P}( J \ \text{BSs in region} \ \mathcal{D}_0) =\frac{(\Lambda_b\alpha)^Je^{-\Lambda_b\alpha}}{\Gamma(J+1)} \ \text{and}\ \text{P}( K \ \text{users in region} \ \mathcal{D}_0) =\frac{(\Lambda_u\alpha)^Ke^{-\Lambda_u\alpha}}{\Gamma(K+1)},
	$$
	 where $\alpha=\pi D^2$. Moreover, the  distance of the $j$-th node from the center is a random variable (RV)  with the probability distribution function (PDF) 
	$$f_1(D;t)=\frac{2x}{D^2}, 0\leq t\leq D.$$ 
	\item[S2.] The network is square with a side length of $2D$, and both components of the 
	$j$-th node's location, $(x_j,y_j)$, are RVs with the truncated normal distribution function $$
	f_2(\mu, {\sigma}, -D, D ; t)=\left\{\begin{array}{cc}
		\displaystyle\frac{\phi\left({\mu}, {\sigma}^2 ; t\right)}{\Phi\left({\mu}, {\sigma}^2 ; D\right)-\Phi\left({\mu}, {\sigma}^2 ; -D\right)}, & -D\leq t\leq D ; \\
		0, & \text{other},
	\end{array}\right.
	$$
	where $\phi(\cdot)$ and $\Phi(\cdot)$ represent the PDF and cumulative distribution function (CDF) of the normal distribution, respectively. The parameters ${\mu}$ and ${\sigma}$ are the mean and variance of the normal distribution.
\end{enumerate}

\begin{figure}[htbp]
	\centering
	\includegraphics[width=\textwidth]{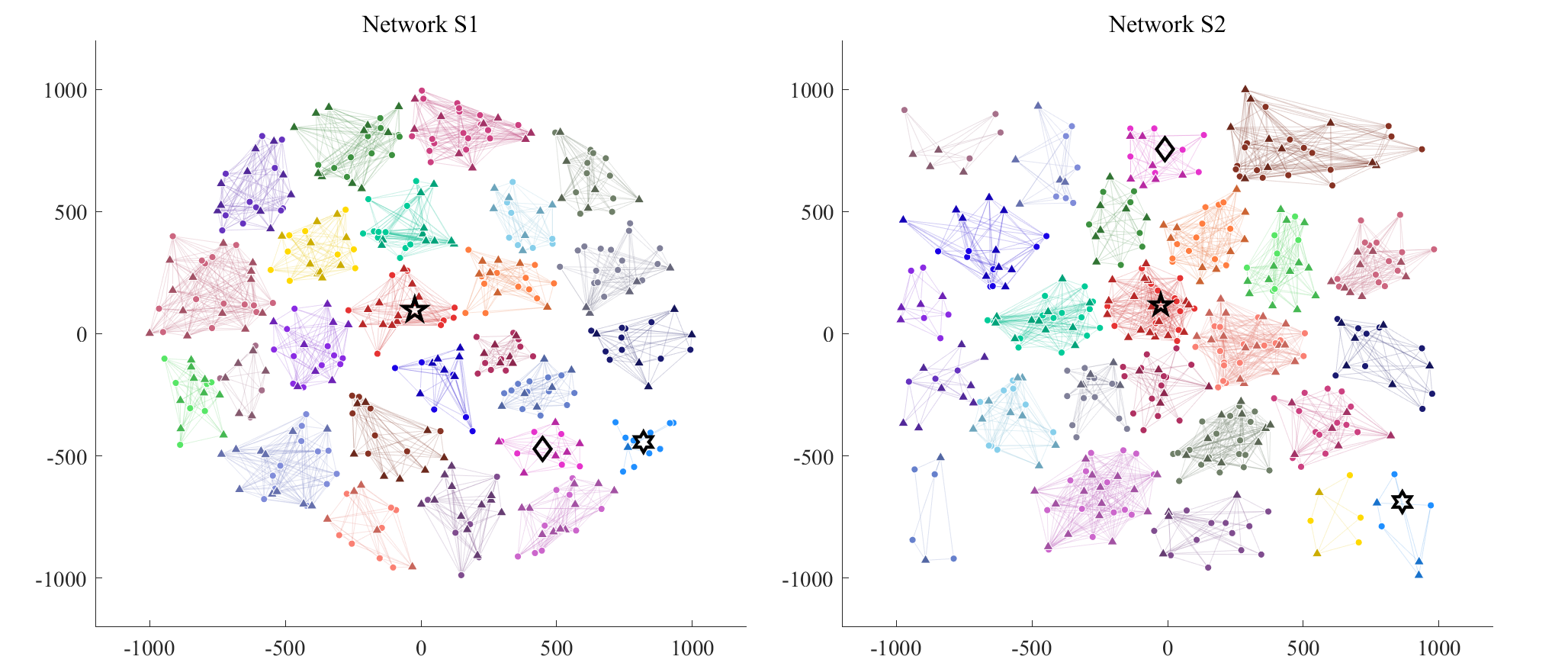}\\
	\caption{Visualization of $\text C^2$ architectures under different network settings. Triangles are users, and circles are BSs. Each color represents  a cluster. The clusters  marked  by the  black pentagon,  the black diamond, and the black hexagon  correspond to the cluster closest to the  network center,  the cluster at the median position\protect\footnotemark, and  the cluster furthest from the network center, respectively.}
	\label{fig_net}
\end{figure}
\footnotetext{Calculate the distances of all clusters to the network center and sort them in descending order,  the cluster corresponding to the median is the one we consider.}

In  network S1, the intensity of BSs  is given by $\Lambda_b=\Omega/(\pi D^2)$, where $\Omega$
is a positive integer, representing the average number of BSs in the region  of $\pi D^2$. With a given $\Lambda_b$, the intensity of users can be calculated according to $\Lambda_u=\Lambda_b\beta$ \cite{yang2022}, where $\beta$ represents the ratio of the number of users to the number of BSs in the entire network. In our two network scenarios, the distribution of users and BSs is exactly the same, so  the $\beta_m$ value in the $m$-th cluster is theoretically equal to the $\beta$ value in the entire network. Due to the randomness of  simulations, $\beta_m$ is around and approximately equal to  $\beta$. In the following experimental analysis, for convenience, we record it as $\beta_m=\beta$.

In network S1, nodes are uniformly distributed, meaning the density of nodes remains constant regardless of the distance from the network center. This is a classical distribution in wireless networks. Conversely, in network S2, nodes follow a truncated normal distribution, which implies that the density off nodes decreases as the distance from the network center increases. This setup mimics real-world scenarios, with a higher concentration of nodes in the central urban area and fewer nodes in the surrounding rural area.

As shown in Fig. \ref{fig_net}, the entire network is partitioned into $M$ non-overlapping  clusters using the K-means algorithm \cite{Lloyd1982}. 
To further validate the generality of our methods, we select three clusters based on their proximity to the network center:  the cluster closest to the network center,  the cluster at the median position, and  the cluster furthest from the network center. The average capacity is then estimated    by averaging over 200 replications for each network scenario.
The setting of basic network parameters  is shown in Table \ref{tab1} below.
\vspace{-1em}
\begin{table}[H]
	\begin{center}
		\caption{The network setting}
		\label{tab1}
		\tabcolsep=1cm
		\begin{tabular}{|c|c|}
			\hline \textbf{Definition and Symbol} & \textbf{Value} \\
			\hline Network scale $(D)$ & $1000 \mathrm{~m}$ \\
			\hline
			Near field threshold $\left(d_{0}\right)$ & $10 \mathrm{~m}$ \\
			\hline
			Far field threshold $\left(d_{1}\right)$ & $50 \mathrm{~m}$ \\
			\hline
			Transmit power $(P)$ & $1 \mathrm{~W}$ \\
			\hline
			Noise power $\left(N_{0}\right)$ & $1 \times 10^{-12} \mathrm{~W}$ \\
			\hline
			Number of clusters $(M)$ & 25 \\
			\hline
			Mean and variance in S2 ($\mu,\sigma^2$) & $(0, 600^2)$\\
			\hline
		\end{tabular}
	\end{center}
\end{table}
\vspace{-1em}

\subsection{Performance evaluation of our methods}

To evaluate the performance of our proposed methods, we compare the capacity  obtained by our approach with those  obtained by TOSE \cite{TOSE2022} and MPM \cite{MPM2022}, with Cholesky decomposition as the baseline. 
According to \cite{TOSE2022}, TOSE requires the selection of the  ratio of spiked eigenvalues  used in the algorithm, which recommends the ratio value set to $0.7$.
In addition, the MPM method in \cite{MPM2022} approximates the channel capacity by an integral form, which requires the choice of the parameter $\eta$ involved in the lower limit.  They suggest that  $4\times 10^{-3}$ is a relatively optimal choice for all cases.
Compared to these two methods, our  methods  offer the advantage of not requiring any parameter tuning.

\begin{figure}[H]
	\centering
	\subfigure{\includegraphics[width=\textwidth]{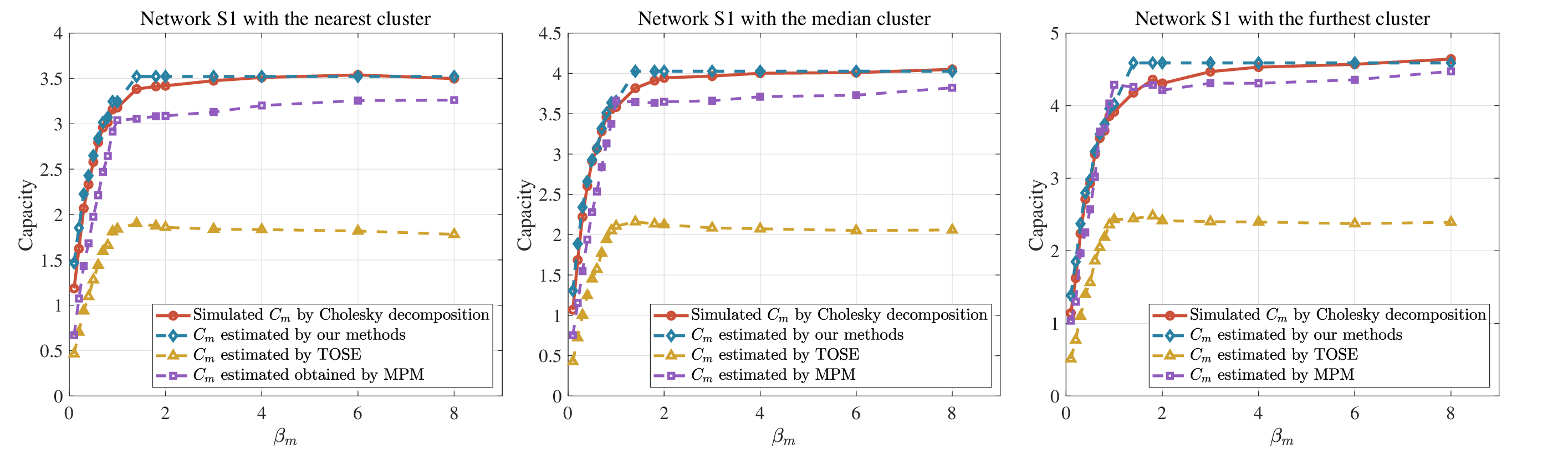}}\\
	\subfigure{\includegraphics[width=\textwidth]{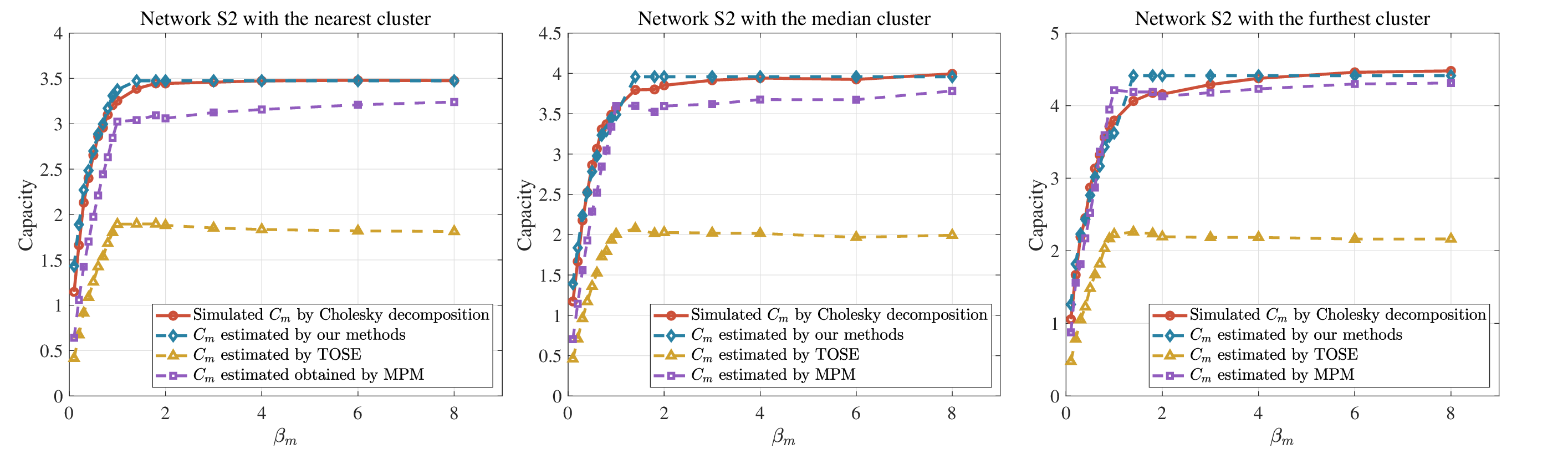}}\\
	\caption{The comparison of the capacity obtained by our methods (blue diamond), TOSE  (yellow triangle), MPM (purple square), and Cholesky decomposition (red circle). The $\beta_m\leq 1$ part of the blue line is obtained by the FISE algorithm, and the  $\beta_m> 1$  part  is calculated by the formula \eqref{limit}.	The first and second rows correspond to the networks S1 and S2. The left, middle, and right columns correspond to the location of  clusters closest to the network center, at the median position, and  furthest from the network center, respectively.}
	\label{fig_capacity1}
\end{figure}
 
Fig. \ref{fig_capacity1} illustrates the capacity comparisons   obtained by different  methods  as $\beta_m$ increases, using Cholesky decomposition  as the baseline for simulating the precise values of
$C_m$ according to equation \eqref{real}.  In network S1,  the intensity of BSs is fixed at  $\Lambda_b=10^{-3}$.	In  network S2, the total number of BSs $J$ across the entire network is set to 5000.
The line with blue diamonds in the figure represents the capacity estimate obtained by our methods.
 Specifically, the capacity for $\beta_m\leq 1$ is obtained using the FISE algorithm. For $\beta_m > 1$, we select the capacity estimate $\widetilde C_m$ when $\beta_m=3$, $\widetilde C_m(3)$, to replace the capacity estimates under all the values of $\beta_m>1$. 

As depicted  in the Fig. \ref{fig_capacity1}, our proposed methods  exhibit superior performance relative to all other methods assessed.  On the one hand, it is observed  that  the FISE method has high accuracy when $\beta_m\leq 1$,  with a relative error in capacity from the baseline results (represented by red circles) of less than 2\%. Conversely, the TOSE method (represented by yellow  triangles) shows significantly poorer performance. 
Although the capacity estimation obtained by MPM (represented by purple squares)  outperforms   TOSE,  it still exhibits greater errors compared to our methods. The errors in MPM stem from the assumption of diagonal noise-plus-interference matrix in the limiting regime, which might not hold accurately for small $\beta_m$, as discussed in the Introduction.
On the other hand, using the capacity estimates $\widetilde C_m(3)$  to replace capacity estimates for all $\beta_m>1$ results in remarkable accuracy across various network settings. As $\beta_m$ increases, the estimation becomes increasingly accurate and  is almost the same as the baseline result, which verifies the validity of Theorem \ref{th1}.

\begin{figure}[ht]
	\centering
	\includegraphics[width=\textwidth]{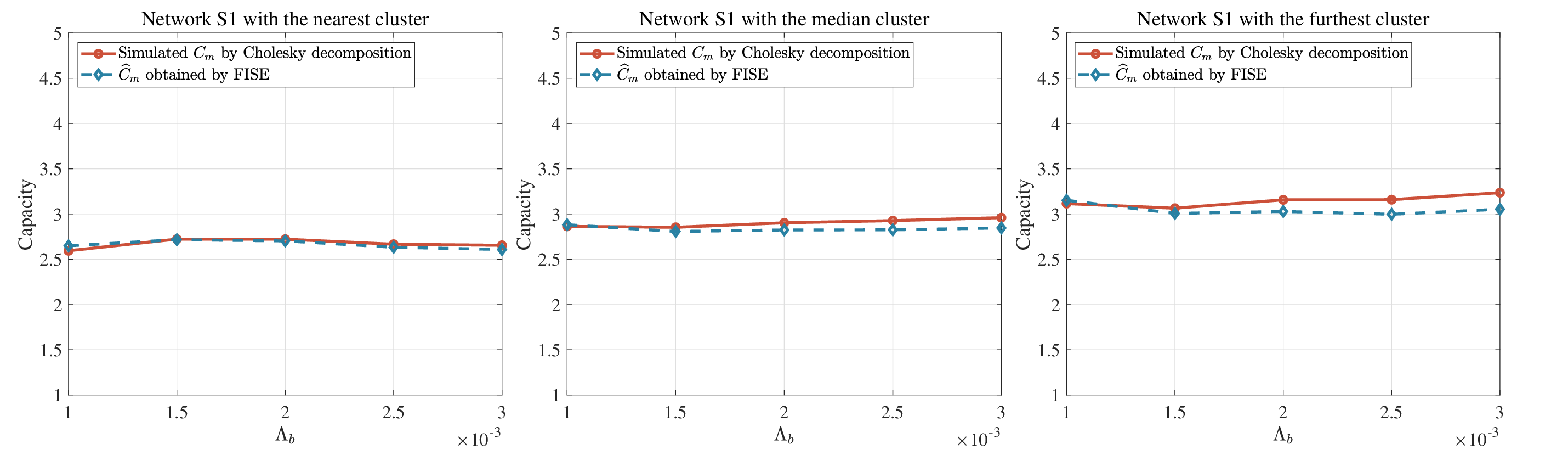}
	\includegraphics[width=\textwidth]{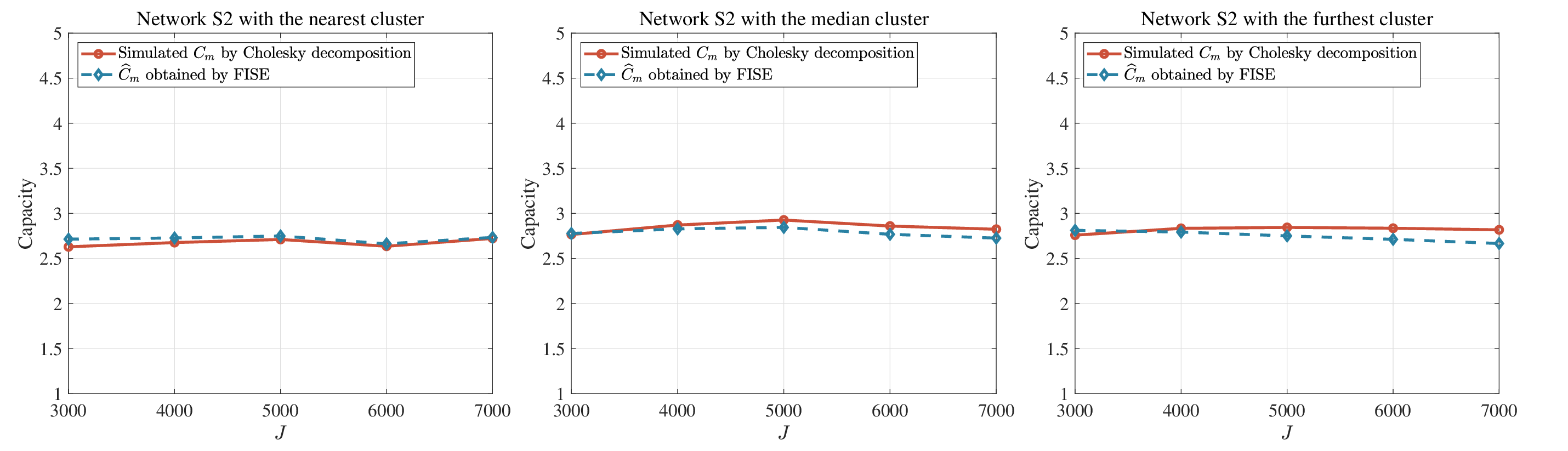}
	\caption{ The comparison of the capacity  calculated by Cholesky decomposition (Solid line) and FISE (dashed line) when $\beta_m=0.5$. The first and second rows correspond to the networks S1 and S2. The left, middle, and right columns correspond to the location of  clusters closest to the network center, at the median position, and  furthest from the network center, respectively. }
	\label{fig_capacity2}
\end{figure}

In Fig. \ref{fig_capacity2}--\ref{fig_capacity3}, we plot to   compare the capacity achieved using our proposed methods versus the Cholesky decomposition method by increasing the density of nodes while maintaining a fixed value of $\beta_m$, equivalently fixed  $\beta$. %, equivalent to fixing values of $\beta$. 
%{\color{red} As aforementioned, $\beta_m$ is equal to $\beta$ theoretically.  Thus, fixed values of $\beta_m$  are equal to fixed values of $\beta$.}
Specifically, in network S1, we augment the node density by varying the intensity of BSs, $\Lambda_b$. Simultaneously, we increase the intensity of users according to $\Lambda_u=\Lambda_b\beta$. For network S2, the density of nodes is escalated by increasing the number of BSs, $J$, while  proportionally increasing the number of users to $K=J\beta$. 
Fig. \ref{fig_capacity2} shows  the capacity comparison  in different scenarios for $\beta_m =0.5$, and Fig. \ref{fig_capacity3} depict similar results for $\beta_m=4$. It can be observed that the capacity estimation remains nearly constant with a fixed $\beta_m$. Our proposed estimates closely align with baseline results across all scenarios, demonstrating high accuracy and robustness suitable for real-world deployments.

\begin{figure}[ht]
	\centering
	\includegraphics[width=\textwidth]{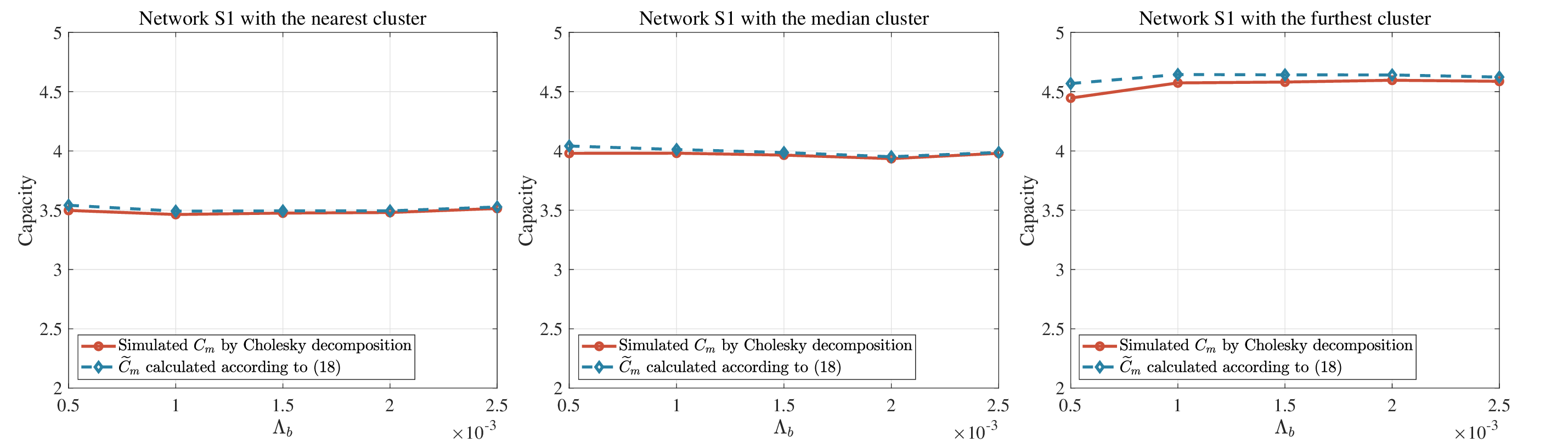}
	\includegraphics[width=\textwidth]{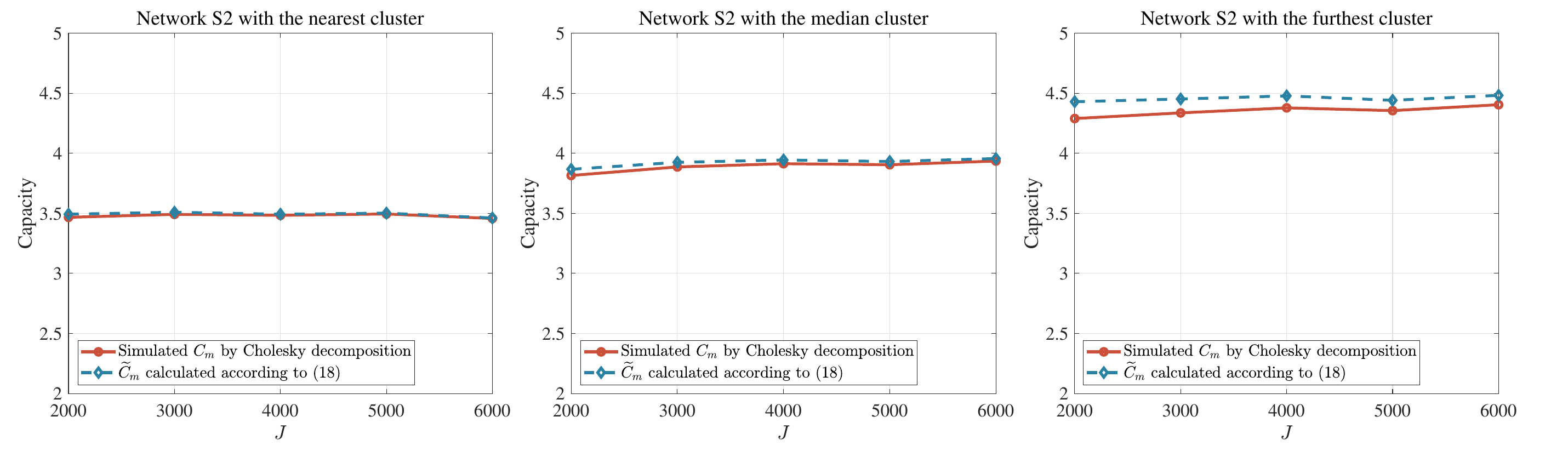}
	\caption{  The comparison of the capacity calculated by Cholesky decomposition (Solid line) and the formula \eqref{limit} (dashed line) when $\beta_m=4$. The first and second rows correspond to the networks S1 and S2. The left, middle, and right columns correspond to the location of  clusters closest to the network center, at the median position, and  furthest from the network center, respectively.}
	\label{fig_capacity3}
\end{figure}
Furthermore, Fig.  \ref{fig_capacity1}--\ref{fig_capacity3}  also demonstrate the generality of our proposed methods. They exhibit high accuracy across different network node distributions,  network area shapes,  cluster  locations, and different values of $\beta_m$. 

\section{Conclusion}\label{isec6}
This paper develops efficient methods for estimating the channel capacity in future ultra-dense wireless networks with random interference, which are fast, accurate and general.  
For $\beta_m\leq 1$, the FISE algorithm is proposed  to realize the fast  estimation of eigenvalues by treating the interference as random. 
According to the analysis, the FISE algorithm has a linear time complexity, much lower than the polynomial time of most existing methods. 
For $\beta_m>1$, we introduce  a computationally simple expression for capacity estimation and prove that it is a stable value that is independent of $\beta_m$. The estimations are also almost the same as the baseline results but with reduced complexity. 
Numerical experiments demonstrate  the high accuracy and superior generality of our proposed methods. Regardless of the network shape,  node  distribution, and cluster location,  our methods provide nearly identical capacity estimations to the baseline method. In future work, we will study the limiting theory of the  Hadamard product between large-dimensional matrices involved in the signal model \eqref{signal}, further achieve even more efficient capacity estimations.

%\section*{Acknowledgments}
%...

%{\appendix[Proof of the Zonklar Equations]
%Use $\backslash$

%{\appendices
%\section*{Proof of the First Zonklar Equation}
%Appendix one text goes here.
% You can choose not to have a title for an appendix if you want by leaving the argument blank
%\section*{Proof of the Second Zonklar Equation}
%Appendix two text goes here.}

 \bibliographystyle{IEEEtran}
 
 \bibliography{IEEE_Fisher}

% Generated by IEEEtran.bst, version: 1.14 (2015/08/26)
\begin{thebibliography}{10}
\providecommand{\url}[1]{#1}
\csname url@samestyle\endcsname
\providecommand{\newblock}{\relax}
\providecommand{\bibinfo}[2]{#2}
\providecommand{\BIBentrySTDinterwordspacing}{\spaceskip=0pt\relax}
\providecommand{\BIBentryALTinterwordstretchfactor}{4}
\providecommand{\BIBentryALTinterwordspacing}{\spaceskip=\fontdimen2\font plus
\BIBentryALTinterwordstretchfactor\fontdimen3\font minus
  \fontdimen4\font\relax}
\providecommand{\BIBforeignlanguage}[2]{{%
\expandafter\ifx\csname l@#1\endcsname\relax
\typeout{** WARNING: IEEEtran.bst: No hyphenation pattern has been}%
\typeout{** loaded for the language `#1'. Using the pattern for}%
\typeout{** the default language instead.}%
\else
\language=\csname l@#1\endcsname
\fi
#2}}
\providecommand{\BIBdecl}{\relax}
\BIBdecl

\bibitem{Shannon1948}
C.~E. Shannon, ``A mathematical theory of communication,'' \emph{The Bell
  System Technical Journal}, vol.~27, pp. 623--656, 1948.

\bibitem{cover2006elements}
T.~M. Cover and J.~A. Thomas, \emph{Elements of Information Theory},
  2nd~ed.\hskip 1em plus 0.5em minus 0.4em\relax John Wiley \& Sons, Ltd, 2006.

\bibitem{TE1999multi}
E.~Telatar, ``Capacity of multi-antenna gaussian channels,'' \emph{European
  Transactions on Telecommunications}, vol.~10, no.~6, pp. 585--595, 1999.

\bibitem{yang2022}
L.~Yang, P.~Li, M.~Dong, B.~Bai, D.~Zaporozhets, X.~Chen, W.~Han, and B.~Li,
  ``C2: A capacity-centric architecture toward future wireless networking,''
  \emph{IEEE Transactions on Wireless Communications}, vol.~21, no.~10, pp.
  8134--8147, 2022.

\bibitem{wang2022rcn}
J.~Wang, L.~Dai, L.~Yang, and B.~Bai, ``Rate-constrained network decomposition
  for clustered cell-free networking,'' in \emph{IEEE International Conference
  on Communications (ICC 2022)}, 2022, pp. 2549--2554.

\bibitem{deng2022cgn}
C.~Deng, L.~Yang, H.~Wu, D.~Zaporozhets, M.~Dong, and B.~Bai, ``{CGN}: A
  capacity-guaranteed network architecture for future ultra-dense wireless
  systems,'' in \emph{IEEE International Conference on Communications (ICC
  2022)}, 2022, pp. 1853--1858.

\bibitem{Dai2017}
L.~Dai and B.~Bai, ``Optimal decomposition for large-scale infrastructure-based
  wireless networks,'' \emph{IEEE Transactions on Wireless Communications},
  vol.~16, no.~8, pp. 4956--4969, 2017.

\bibitem{6G2021}
M.~Matthaiou, O.~Yurduseven, H.~Q. Ngo, D.~Morales-Jimenez, S.~L. Cotton, and
  V.~F. Fusco, ``The road to 6g: Ten physical layer challenges for
  communications engineers,'' \emph{IEEE Communications Magazine}, vol.~59,
  no.~1, pp. 64--69, 2021.

\bibitem{wang2023}
J.~Wang, L.~Dai, L.~Yang, and B.~Bai, ``Clustered cell-free networking: A graph
  partitioning approach,'' \emph{IEEE Transactions on Wireless Communications},
  vol.~22, no.~8, pp. 5349--5364, 2023.

\bibitem{TuVer2004}
A.~Tulino and S.~Verdú, ``Random matrix theory and wireless communications,''
  \emph{Foundations and Trends in Communications and Information Theorys},
  vol.~1, no.~1, 2004.

\bibitem{TOSE2022}
D.~Jiang, L.~Yang, H.~Hao, and R.~Wang, ``{TOSE}: A fast capacity estimation
  algorithm based on spike approximations,'' in \emph{2022 IEEE 96th Vehicular
  Technology Conference (VTC 2022-Fall)}, 2022, pp. 1--6.

\bibitem{MPM2022}
H.~Hao, D.~Jiang, L.~Yang, H.~Wu, and B.~Bai, ``The moment passing method for
  wireless channel capacity estimation,'' in \emph{2022 IEEE Global
  Communications Conference (GLOBECOM 2022)}, 2022, pp. 3605--3610.

\bibitem{couillet_liao_2022}
R.~Couillet and Z.~Liao, \emph{Random Matrix Methods for Machine
  Learning}.\hskip 1em plus 0.5em minus 0.4em\relax Cambridge University Press,
  2022.

\bibitem{How2016}
G.~Interdonato, H.~Q. Ngo, E.~G. Larsson, and P.~Frenger, ``How much do
  downlink pilots improve cell-free massive mimo?'' in \emph{2016 IEEE Global
  Communications Conference (GLOBECOM)}, 2016, pp. 1--7.

\bibitem{Bashar2019}
M.~Bashar, K.~Cumanan, A.~G. Burr, M.~Debbah, and H.~Q. Ngo, ``On the uplink
  max–min sinr of cell-free massive mimo systems,'' \emph{IEEE Transactions
  on Wireless Communications}, vol.~18, no.~4, pp. 2021--2036, 2019.

\bibitem{5G2020}
S.~Buzzi, C.~D’Andrea, A.~Zappone, and C.~D’Elia, ``User-centric 5g
  cellular networks: Resource allocation and comparison with the cell-free
  massive mimo approach,'' \emph{IEEE Transactions on Wireless Communications},
  vol.~19, no.~2, pp. 1250--1264, 2020.

\bibitem{David2005}
D.~Tse and P.~Viswanath, \emph{Fundamentals of Wireless Communication}.\hskip
  1em plus 0.5em minus 0.4em\relax Cambridge University Press, 2005.

\bibitem{bai2010spectral}
Z.~Bai and J.~W. Silverstein, \emph{Spectral Analysis of Large Dimensional
  Random Matrices}.\hskip 1em plus 0.5em minus 0.4em\relax Springer, 2010.

\bibitem{Girko2001}
V.~L. Girko, \emph{Theory of Stochastic Canonical Equations}.\hskip 1em plus
  0.5em minus 0.4em\relax Springer Science \& Business Media, 2001.

\bibitem{Hachem2007}
W.~Hachem, P.~Loubaton, and J.~Najim, ``{Deterministic equivalents for certain
  functionals of large random matrices},'' \emph{The Annals of Applied
  Probability}, vol.~17, no.~3, pp. 875 -- 930, 2007.

\bibitem{Hachem2008}
------, ``{A CLT for information-theoretic statistics of Gram random matrices
  with a given variance profile},'' \emph{The Annals of Applied Probability},
  vol.~18, no.~6, pp. 2071 -- 2130, 2008.

\bibitem{Silver2023}
J.~W. Silverstein, ``Limiting eigenvalue behavior of a class of large
  dimensional random matrices formed from a hadamard product,'' \emph{Random
  Matrices: Theory and Applications}, vol.~12, no.~01, p. 2250050, 2023.

\bibitem{wang2016}
Q.~Wang and J.~Yao, ``{Extreme eigenvalues of large-dimensional spiked Fisher
  matrices with application},'' \emph{The Annals of Statistics}, vol.~45,
  no.~1, pp. 415 -- 460, 2017.

\bibitem{Jiang2021}
D.~Jiang, Z.~Hou, and J.~Hu, ``The limits of the sample spiked eigenvalues for
  a high-dimensional generalized fisher matrix and its applications,''
  \emph{Journal of Statistical Planning and Inference}, vol. 215, pp. 208--217,
  2021.

\bibitem{xie2021}
J.~Xie, Y.~Zeng, and L.~Zhu, ``Limiting laws for extreme eigenvalues of
  large-dimensional spiked fisher matrices with a divergent number of spikes,''
  \emph{J. Multivar. Anal.}, vol. 184, no.~C, jul 2021.

\bibitem{zheng2023}
Y.~Zeng and L.~Zhu, ``Order determination for spiked-type models with a
  divergent number of spikes,'' \emph{Computational Statistics \& Data
  Analysis}, vol. 182, p. 107704, 2023.

\bibitem{Khan2015}
F.~A. Khan, H.~He, J.~Xue, and T.~Ratnarajah, ``Performance analysis of cloud
  radio access networks with distributed multiple antenna remote radio heads,''
  \emph{IEEE Transactions on Signal Processing}, vol.~63, no.~18, pp.
  4784--4799, 2015.

\bibitem{Lloyd1982}
S.~Lloyd, ``Least squares quantization in pcm,'' \emph{IEEE Transactions on
  Information Theory}, vol.~28, no.~2, pp. 129--137, 1982.

\end{thebibliography}
 % argument is your BibTeX string definitions and bibliography database(s)
%\bibliography{IEEEabrv,../bib/paper}
%

\end{document}